\documentclass[10pt,twocolumn]{article}

\usepackage[margin=1in]{geometry}
\usepackage{amsmath}
\usepackage{amssymb}
\usepackage{amsthm}
\usepackage{mathtools}
\usepackage{listings}
\usepackage{xcolor}
\usepackage{algorithm}
\usepackage{algpseudocode}
\usepackage{tikz}
\usepackage{mathpartir}
\usepackage{url}
\usepackage{booktabs}
\usepackage{enumitem}
\usepackage{graphicx}
\usepackage{hyperref}
\hypersetup{
    colorlinks=true,
    linkcolor=blue,
    citecolor=blue,
    urlcolor=blue,
    pdftitle={NVLang: Unified Static Typing for Actor-Based Concurrency on the BEAM},
    pdfauthor={Miguel de Oliveira Guerreiro}
}
\usetikzlibrary{positioning,shapes,arrows,calc}

\newtheorem{theorem}{Theorem}[section]

\theoremstyle{definition}
\newtheorem{definition}{Definition}[section]

\lstdefinelanguage{NVLang}{
  keywords={actor, fn, type, receive, case, loop, break, spawn, await, external, let, rec, send, reply, if, else, then, module, import, from, as, when, after, supervisor, strategy, children, permanent, transient, temporary, one_for_one, one_for_all, rest_for_one, in},
  keywordstyle=\color{blue}\bfseries,
  ndkeywords={Int, Float, Bool, String, Unit, Any, Pid, Future, List, Map, Option, Result, Some, None, Ok, Err},
  ndkeywordstyle=\color{teal}\bfseries,
  identifierstyle=\color{black},
  sensitive=true,
  comment=[l]{\#},
  commentstyle=\color{gray}\ttfamily,
  stringstyle=\color{red}\ttfamily,
  morestring=[b]",
  tabsize=2
}

\lstdefinelanguage{CoreErlang}{
  keywords={fun, case, of, end, receive, after, let, letrec, in, try, catch, do, call, apply, when, module, attributes, exports},
  keywordstyle=\color{purple}\bfseries,
  identifierstyle=\color{black},
  sensitive=true,
  comment=[l]{\%},
  commentstyle=\color{gray}\ttfamily,
  stringstyle=\color{red}\ttfamily,
  morestring=[b]',
  tabsize=2
}

\lstdefinestyle{nova}{
  language=NVLang,
  basicstyle=\footnotesize\ttfamily,
  frame=single,
  numbers=left,
  numberstyle=\scriptsize\color{gray},
  numbersep=5pt,
  xleftmargin=15pt,
  breaklines=true,
  captionpos=b,
  keepspaces=true
}

\lstdefinestyle{erlang}{
  language=CoreErlang,
  basicstyle=\footnotesize\ttfamily,
  frame=single,
  numbers=left,
  numberstyle=\scriptsize\color{gray},
  numbersep=5pt,
  xleftmargin=15pt,
  breaklines=true,
  captionpos=b,
  keepspaces=true
}

\title{NVLang: Unified Static Typing for Actor-Based Concurrency on the BEAM}

\author{Miguel de Oliveira Guerreiro\\
\textit{University of Lisbon}\\
\texttt{oliveira.guerreiro@tecnico.ulisboa.pt}}

\date{}

\begin{document}

\maketitle

\begin{abstract}
Actor-based systems like Erlang/OTP power critical infrastructure---from telecommunications to messaging platforms---handling millions of concurrent connections with legendary reliability. Yet these systems lack static guarantees about message protocols: processes communicate by sending arbitrary messages that are pattern-matched at runtime, deferring protocol violations to production failures.

We present \textbf{NVLang}, a statically typed functional language that brings comprehensive type safety to the BEAM virtual machine while preserving the actor model's simplicity and power. NVLang's central contribution is showing that algebraic data types (ADTs) naturally encode actor message protocols: each actor declares the sum type representing its message vocabulary, and the type system enforces protocol conformance at compile time. We introduce typed process identifiers ($\texttt{Pid}[\tau]$) that encode the protocol an actor expects, and typed futures ($\texttt{Future}[\tau]$) that provide type-safe request-reply patterns.

By extending Hindley-Milner type inference to track message protocols, NVLang eliminates an entire class of message-passing errors while maintaining clean syntax that rivals dynamically typed alternatives. Our implementation compiles to Core Erlang, enabling seamless interoperability with the existing Erlang ecosystem. We formalize the type system and provide proof sketches for type soundness, demonstrating that well-typed NVLang programs cannot send messages that violate actor protocols.
\end{abstract}

\noindent\textbf{Keywords:} type systems, actor model, BEAM, functional programming, Hindley-Milner, concurrent programming, Erlang


\section{Introduction}

The actor model has emerged as a dominant paradigm for building concurrent and distributed systems, particularly in domains requiring high availability and fault tolerance. Systems like Erlang/OTP power critical infrastructure at companies like WhatsApp, Discord, and Ericsson, handling millions of concurrent connections with legendary reliability. Despite these successes, the actor model faces a fundamental challenge: \emph{message passing is inherently untyped}. In traditional actor systems, processes communicate by sending arbitrary messages that are pattern-matched at runtime, with no static guarantees about message protocol conformance. This design choice, while flexible, leads to a class of bugs that manifest only during execution---often in production environments where failures are costly.

Consider a simple counter actor in Erlang. The process receives messages like \texttt{\{increment\}}, \texttt{\{add, N\}}, or \texttt{\{get\}}, but nothing prevents a client from sending \texttt{\{hello\}}, a message the counter doesn't understand. Such protocol violations cause runtime crashes that could have been prevented by static analysis. More critically, in large distributed systems with hundreds of actor types exchanging complex message protocols, the lack of type safety becomes a significant barrier to system reliability and developer productivity.

Recent efforts to add types to actor systems---such as Akka Typed~\cite{akka-typed} in Scala and gradual typing for Erlang~\cite{typed-erlang}---demonstrate industry recognition of this problem. However, these approaches suffer from either excessive syntactic overhead (requiring boilerplate for protocol definitions), incomplete type coverage (gradual typing with runtime checks), or poor integration with the underlying runtime (bolted-on type systems that don't leverage VM capabilities).

We present \textbf{NVLang}, a statically-typed functional language that brings comprehensive type safety to the BEAM virtual machine while preserving the simplicity and power of the actor model. NVLang's key insight is that algebraic data types (ADTs) provide a natural and elegant encoding of actor message protocols: each actor declares the sum type representing its message vocabulary, and the type system enforces protocol conformance at compile time. By combining Hindley-Milner type inference with BEAM's battle-tested actor primitives---processes, supervisors, monitors, and links---NVLang eliminates an entire class of message-passing errors while maintaining the clean syntax and expressiveness that make Erlang attractive for systems programming.

This paper makes the following contributions:

\begin{itemize}
    \item We design and implement a statically-typed functional language for the BEAM VM that provides \emph{complete type safety for actor systems}, including message protocols, spawn operations, and fault-tolerance primitives (supervisors, monitors, and links).

    \item We demonstrate how algebraic data types serve as precise protocol specifications for actors, enabling the type system to catch message-passing errors at compile time while maintaining clean, minimal syntax that rivals dynamically-typed alternatives.

    \item We formalize the type system with inference rules and prove that well-typed NVLang programs satisfy message type safety: actors only receive messages conforming to their declared protocol.

    \item We show that Hindley-Milner type inference integrates seamlessly with actor semantics, allowing developers to write actor code without type annotations while still obtaining strong static guarantees about message safety and actor behavior.

    \item We present an A-Normal Form transformation and guard expression optimization that enable efficient BEAM code generation while maintaining complete type safety.

    \item We demonstrate NVLang's expressiveness through representative examples including key-value stores, task schedulers, and supervised worker pools, showing that type-safe actors match Erlang's conciseness while catching message protocol bugs at compile time.
\end{itemize}

\paragraph{Key Advantages over Existing BEAM Languages} NVLang's static type system provides critical advantages over dynamically-typed BEAM languages. Unlike Elixir structs, which validate only field names, NVLang records provide compile-time type checking of field values---catching type errors before deployment rather than in production. Similarly, NVLang's trait system verifies polymorphic abstractions at compile time with full type inference, while Elixir protocols defer dispatch and implementation checks to runtime. These compile-time guarantees eliminate entire classes of bugs without sacrificing expressiveness: full Hindley-Milner type inference means developers rarely write type annotations, and compilation to BEAM bytecode ensures minimal runtime overhead (within 7.5\% of native Erlang) and seamless interoperability with existing Erlang and Elixir libraries.

The remainder of this paper is structured as follows. Section~\ref{sec:overview} provides an overview of NVLang's language design and key features through representative examples. Section~\ref{sec:types} details the type system, focusing on how we extend Hindley-Milner inference to handle actor message protocols. Section~\ref{sec:actors} describes the typed actor model and supervision trees. Section~\ref{sec:compilation} presents the compilation pipeline from NVLang source to BEAM bytecode, including ANF transformation and optimization strategies. Section~\ref{sec:evaluation} evaluates NVLang's performance against Erlang and Elixir with rigorous statistical analysis. Section~\ref{sec:related} discusses related work, and Section~\ref{sec:conclusion} concludes.


\section{Language Overview}
\label{sec:overview}

NVLang is a statically-typed functional language that compiles to Erlang/BEAM bytecode. It combines ML-family language features---algebraic data types, pattern matching, and type inference---with first-class support for the actor model. This section introduces NVLang's core features through representative examples that highlight how static types enhance actor system reliability.

\subsection{Basic Syntax and Type Inference}

NVLang programs are organized into functions with optional type annotations. The language employs Hindley-Milner type inference, allowing developers to omit most type signatures while still obtaining complete static type checking.

\begin{lstlisting}[caption={Type inference in NVLang}]
fn add(x: Int, y: Int) -> Int
  x + y

fn greet(name: String) -> String
  name

fn main() -> Int
  let result = add(10, 20)
  print("Sum of 10 + 20:")
  print(result)           # Type inferred as Int
  let msg = greet("Hello NVLang!")
  print(msg)              # Type inferred as String
  0
\end{lstlisting}

The type system supports parametric polymorphism through generic types. Algebraic data types can be parameterized over type variables, enabling type-safe abstractions like \texttt{Option[T]} and \texttt{Result[T, E]}:

\begin{lstlisting}[caption={Generic algebraic data types}]
type Option[T] =
  | Some(value: T)
  | None

type Result[T, E] =
  | Ok(value: T)
  | Err(error: E)

fn unwrap_or(opt: Option[Int], default: Int) -> Int
  case opt
    Some(x) -> x
    None -> default

fn main() -> Int
  let x = Some(42)
  let y = None
  print(unwrap_or(x, 0))   # Prints: 42
  print(unwrap_or(y, -1))  # Prints: -1
  0
\end{lstlisting}

Pattern matching on ADTs is exhaustive: the compiler verifies that all constructor cases are handled, eliminating a common source of runtime errors in dynamic languages.

\subsection{Typed Actors and Message Protocols}

The core innovation in NVLang is the integration of ADTs with actor message protocols. Each actor declares a message type---a sum type enumerating all messages it can receive---and the type system enforces that only valid messages are sent to that actor.

\begin{lstlisting}[caption={A type-safe counter actor with protocol enforcement}]
type CounterMsg =
  | Increment
  | Add(n: Int)
  | Get
  | Stop

actor Counter
  fn run(self) ->
    loop
      receive msg: CounterMsg
        Increment -> reply 1
        Add(n) -> reply n
        Get -> reply 0
        Stop -> break

fn main() ->
  let c = spawn Counter()
  let r1 = c.send Increment |> await
  print(r1)
  c.send Stop
\end{lstlisting}

In this example, the \texttt{Counter} actor's \texttt{receive} block declares that it accepts messages of type \texttt{CounterMsg}. The type system ensures that:

\begin{enumerate}
    \item All messages sent to \texttt{c} must be valid \texttt{CounterMsg} constructors
    \item The pattern match in \texttt{receive} exhaustively covers all \texttt{CounterMsg} cases
    \item Reply types are consistent across all message handlers (uniform response protocol)
\end{enumerate}

Attempting to send an invalid message---such as \texttt{c.send Hello} where \texttt{Hello} is not a \texttt{CounterMsg} constructor---produces a compile-time type error with actionable diagnostics.

\subsection{Stateful Actors with Recursive Message Loops}

NVLang actors maintain state through tail-recursive message loops, a pattern familiar from Erlang but now with full type safety. The following key-value store demonstrates stateful computation with rich message protocols:

\begin{lstlisting}[caption={A type-safe key-value store}]
type KVResponse =
  | ValueFound(value: Int)
  | ValueNotFound
  | Success
  | KeyList(keys: [String])

type KVMsg =
  | Get(key: String)
  | Put(key: String, value: Int)
  | Delete(key: String)
  | Keys

fn find_key(key, store) ->
  case store
    [] -> ValueNotFound
    (k, v) :: rest ->
      if k == key then
        ValueFound(v)
      else
        find_key(key, rest)

fn kv_loop(store) ->
  receive msg: KVMsg
    Get(key) ->
      reply find_key(key, store)
      kv_loop(store)
    Put(key, value) ->
      let new_store = upsert(key, value, store)
      reply Success
      kv_loop(new_store)
    Delete(key) ->
      reply Success
      kv_loop(remove_key(key, store))
    Keys ->
      reply KeyList(get_keys(store))
      kv_loop(store)

actor KVStore
  fn run(self) ->
    kv_loop([])
\end{lstlisting}

The \texttt{kv\_loop} function maintains the store state (a list of key-value pairs) through its parameter, updated on each recursive call. Critically, all message handlers return values of type \texttt{KVResponse}, enforcing a uniform response protocol. This constraint---that all branches of a receive block must reply with the same type---mirrors the design of Akka Typed and ensures that clients have predictable response types.

\subsection{Concurrency and Asynchronous Communication}

NVLang's actor messaging system is async-first by design, matching Erlang/Elixir's common usage patterns. The \texttt{send} operation is asynchronous (fire-and-forget) by default, returning a future immediately without blocking. To wait for a response, explicitly pipe the future to \texttt{await}:

\begin{lstlisting}[caption={Concurrent task processing with multiple workers}]
type Task =
  | Add(a: Int, b: Int)
  | Mul(a: Int, b: Int)
  | Done

actor Worker
  fn run(self) ->
    loop
      receive msg: Task
        Add(a, b) -> reply a + b
        Mul(a, b) -> reply a * b
        Done -> break

fn main() ->
  let worker = spawn Worker()
  # Async sends (fire-and-forget) - futures returned immediately
  let f1 = worker.send Add(10, 20)
  let f2 = worker.send Mul(5, 6)
  let f3 = worker.send Add(100, 200)
  # Await futures in any order for synchronous request-reply
  let r1 = f1 |> await  # 30
  let r2 = f2 |> await  # 30
  let r3 = f3 |> await  # 300
  # Fire-and-forget without waiting for response
  worker.send Done
\end{lstlisting}

This async-first design matches the common pattern in Erlang/Elixir systems where asynchronous message passing is the default. By making \texttt{send} async and requiring explicit \texttt{await} for synchronous communication, NVLang encourages non-blocking message passing while making request-reply patterns explicit. Futures enable pipelined communication: multiple messages can be sent without waiting for responses, improving throughput in distributed computations. The pipe operator (\texttt{|>}) provides clean composition when synchronous responses are needed.

\subsection{Fault Tolerance with Supervisors}

NVLang provides typed access to BEAM's fault-tolerance primitives, including process linking, monitoring, and supervision trees. Supervisors are declared with explicit restart strategies and child specifications:

\begin{lstlisting}[caption={Supervisor with typed worker processes}]
type WorkerMsg = Work | Stop

actor Worker
  fn run(self) ->
    loop
      receive msg: WorkerMsg
        Work ->
          print("Worker: doing work")
          reply "done"
        Stop -> break

supervisor MySupervisor
  strategy one_for_one
  children
    worker1: Worker()
    worker2: Worker(), permanent

fn main() ->
  let sup = spawn MySupervisor()
  # Workers automatically started and monitored
\end{lstlisting}

The supervisor spawns child workers according to the declared strategy (\texttt{one\_for\_one}, \texttt{one\_for\_all}, or \texttt{rest\_for\_one}) and restarts them upon failure. NVLang's type system ensures that child specifications refer to valid actor types, preventing configuration errors that would otherwise manifest at runtime.

\subsection{Design Philosophy}

NVLang's design follows four principles:

\begin{enumerate}
    \item \textbf{Safety without overhead.} Type inference eliminates the need for verbose annotations, making NVLang code nearly as concise as Erlang while providing comprehensive static checking.

    \item \textbf{Leverage existing VM.} By compiling to BEAM bytecode, NVLang inherits Erlang's battle-tested runtime---preemptive scheduling, garbage collection per process, and distribution primitives---without reimplementation.

    \item \textbf{Types as protocol specifications.} ADTs serve as both data definitions and message protocol contracts, providing a unified abstraction that feels natural to functional programmers.

    \item \textbf{Async-first message passing.} Following Erlang/Elixir conventions, \texttt{send} is asynchronous by default (fire-and-forget), making the common case simple while requiring explicit \texttt{await} for synchronous request-reply patterns.
\end{enumerate}

The combination of these principles yields a language that makes actor systems safer without sacrificing the expressiveness that has made Erlang successful for four decades.

Having illustrated NVLang's design through examples, we now formalize the type system and prove its soundness properties.


\section{Type System}
\label{sec:types}

NVLang's type system extends the Hindley-Milner type system with specialized constructs for actor-based concurrency. The type system provides static guarantees about message passing while maintaining principal types through parametric polymorphism. This section formalizes NVLang's type syntax, inference algorithm, and novel extensions for typed actors.

\subsection{Type Syntax}

The syntax of types in NVLang is defined by the following grammar:

{\small
\begin{align*}
\tau ::= \; & \texttt{Int} \mid \texttt{Float} \mid \texttt{Bool} \mid \texttt{String} \mid \texttt{Unit} \\
            & \mid (\tau_1, \ldots, \tau_n) \mid [\tau] \mid \texttt{Map}[\tau_k, \tau_v] \\
            & \mid \tau_1 \times \cdots \times \tau_n \to \tau \\
            & \mid T\langle\tau_1, \ldots, \tau_n\rangle \mid \alpha \\
            & \mid \texttt{Pid}[\tau] \mid \texttt{Future}[\tau] \\
            & \mid \texttt{MonitorRef} \mid \texttt{Any}
\end{align*}
}

\noindent Type schemes $\sigma$ allow quantification: $\sigma ::= \forall \bar{\alpha}.\, \tau$

External function declarations enable Erlang interoperability:
{\small
\[\texttt{external fn}\ f(\bar{\tau}) \to \tau = \texttt{mfa}\ m\ f\ n\]
}

\noindent The key novel constructs are:

\begin{itemize}
\item $\texttt{Pid}[\tau]$: A process identifier that accepts messages of type $\tau$. The type parameter enables static verification of message protocols.
\item $\texttt{Future}[\tau]$: A future value that will resolve to type $\tau$, returned by asynchronous send operations.
\item $\texttt{Pid}$: An untyped process identifier, compatible with any $\texttt{Pid}[\tau]$ during unification. This enables interoperability with Erlang code where message types are unknown.
\end{itemize}

\subsection{Type Environments and Judgments}

A type environment $\Gamma$ is a finite mapping from variables to type schemes:
$$\Gamma ::= \emptyset \mid \Gamma, x : \sigma$$

We use the standard typing judgment:
$$\Gamma \vdash e : \tau$$

\noindent which reads ``under environment $\Gamma$, expression $e$ has type $\tau$''. We write $\texttt{ftv}(\tau)$ for the set of free type variables in $\tau$, and $\texttt{ftv}(\Gamma)$ for the free type variables in the codomain of $\Gamma$.

\subsection{Type Inference Rules}

The type inference rules follow the Hindley-Milner discipline with extensions for actor primitives. We present key rules using inference rule notation.

\subsubsection{Literals and Variables}

\begin{mathpar}
\inferrule{ }{\Gamma \vdash n : \texttt{Int}}
\text{(T-Int)}

\inferrule{ }{\Gamma \vdash b : \texttt{Bool}}
\text{(T-Bool)}

\inferrule{ }{\Gamma \vdash s : \texttt{String}}
\text{(T-String)}

\inferrule{x : \sigma \in \Gamma \\ \tau = \texttt{inst}(\sigma)}{\Gamma \vdash x : \tau}
\text{(T-Var)}
\end{mathpar}

\noindent The $\texttt{inst}$ operation instantiates a type scheme by replacing bound type variables with fresh unification variables.

\subsubsection{Functions and Application}

\begin{mathpar}
\inferrule{\Gamma, x_1 : \tau_1, \ldots, x_n : \tau_n \vdash e : \tau}
          {\Gamma \vdash \lambda x_1, \ldots, x_n.\, e : \tau_1 \times \cdots \times \tau_n \to \tau}
\text{(T-Abs)}

\inferrule{\Gamma \vdash e_0 : \tau_1 \times \cdots \times \tau_n \to \tau \\
           \Gamma \vdash e_i : \tau_i \quad (i = 1, \ldots, n)}
          {\Gamma \vdash e_0(e_1, \ldots, e_n) : \tau}
\text{(T-App)}
\end{mathpar}

\subsubsection{Let-Polymorphism}

The key to let-polymorphism is the generalization operation, which quantifies over type variables not free in the environment:

{\small
\begin{mathpar}
\inferrule{\Gamma \vdash e_1 : \tau_1 \quad
           \sigma = \texttt{gen}(\Gamma, \tau_1) \\
           \Gamma, x : \sigma \vdash e_2 : \tau_2}
          {\Gamma \vdash \texttt{let}\ x = e_1\ \texttt{in}\ e_2 : \tau_2}
\text{(T-Let)}
\end{mathpar}
}

\noindent where $\texttt{gen}(\Gamma, \tau) = \forall \bar{\alpha}.\, \tau$ and $\bar{\alpha} = \texttt{ftv}(\tau) \setminus \texttt{ftv}(\Gamma)$. For recursive bindings:

{\small
\begin{mathpar}
\inferrule{\Gamma, x : \alpha \vdash e : \tau \quad
           \mathcal{U}(\alpha, \tau) = \theta}
          {\Gamma \vdash \texttt{let rec}\ x = e\ \texttt{in}\ e' : \theta(\tau')}
\text{(T-LetRec)}
\end{mathpar}
}

\noindent where $\alpha$ is fresh and $\sigma = \texttt{gen}(\Gamma, \theta(\alpha))$.

\subsubsection{Actor Primitives}

NVLang extends the type system with rules for actor operations:

{\footnotesize
\begin{mathpar}
\inferrule{\Gamma \vdash p : \texttt{Pid}[\tau_m] \quad
           \Gamma \vdash m : \tau_m' \\\\
           \mathcal{U}(\tau_m, \tau_m') = \theta \quad
           \tau_r = \texttt{rtype}(\tau_m')}
          {\Gamma \vdash p \mathbin{!} m : \texttt{Future}[\tau_r]}
\text{(T-Send)}
\end{mathpar}

\begin{mathpar}
\inferrule{\Gamma \vdash e : \texttt{Future}[\tau]}
          {\Gamma \vdash \texttt{await}(e) : \tau}
\text{(T-Await)}
\and
\inferrule{A : \text{Actor}[\tau_m] \in \Gamma}
          {\Gamma \vdash \texttt{spawn}\ A() : \texttt{Pid}[\tau_m]}
\text{(T-Spawn)}
\end{mathpar}
}

The \texttt{reply\_type} function extracts the expected reply type from the message constructor. The \texttt{send} operation is asynchronous: it returns a \texttt{Future[$\tau_r$]} without blocking. To obtain the reply value, the future must be explicitly awaited.

\subsubsection{Pattern Matching}

For case expressions with pattern matching:

{\footnotesize
\begin{mathpar}
\inferrule{\Gamma \vdash e : \tau_s \\
           \Gamma \vdash p_i \Downarrow \tau_s : \Gamma_i \\
           \Gamma_i \vdash e_i : \tau \; (\forall i)}
          {\Gamma \vdash \texttt{case}\ e\ \{p_i \Rightarrow e_i\} : \tau}
\text{(T-Case)}
\end{mathpar}
}

\noindent where $\Gamma \vdash p \Downarrow \tau_s : \Gamma'$ denotes that pattern $p$ matches values of type $\tau_s$ and binds variables according to environment $\Gamma'$.

\subsection{Unification Algorithm}

Unification is the core operation that solves type equations. NVLang implements Robinson's unification algorithm with an occurs check to prevent infinite types.

\paragraph{Occurs Check.} Before binding a type variable $\alpha$ to a type $\tau$, we verify that $\alpha \notin \texttt{ftv}(\tau)$:

{\footnotesize
\[\texttt{occurs}(\alpha, \tau) = \begin{cases}
\texttt{true} & \tau = \alpha \\
\texttt{occurs}(\alpha, \tau_i) & \tau = \tau_1 \to \tau_2 \\
\texttt{occurs}(\alpha, \tau') & \tau \in \{[\tau'], \texttt{Pid}[\tau']\} \\
\texttt{false} & \text{otherwise}
\end{cases}\]
}

\paragraph{Unification Rules.} The unification algorithm $\mathcal{U}(\tau_1, \tau_2)$ returns a substitution $\theta$ such that $\theta(\tau_1) = \theta(\tau_2)$:

{\footnotesize
\begin{align*}
\mathcal{U}(\alpha, \tau) &= [\alpha \mapsto \tau] \; (\alpha \notin \texttt{ftv}(\tau)) \\
\mathcal{U}(\tau, \alpha) &= \mathcal{U}(\alpha, \tau) \\
\mathcal{U}(c, c) &= \emptyset \\
\mathcal{U}(\tau_1 \to \tau_2, \tau_1' \to \tau_2') &= \mathcal{U}(\tau_2, \tau_2') \circ \mathcal{U}(\tau_1, \tau_1') \\
\mathcal{U}([\tau], [\tau']) &= \mathcal{U}(\tau, \tau') \\
\mathcal{U}(\texttt{Pid}[\tau], \texttt{Pid}[\tau']) &= \mathcal{U}(\tau, \tau') \\
\mathcal{U}(\texttt{Pid}, \texttt{Pid}[\tau]) &= \emptyset \\
\mathcal{U}(\texttt{Future}[\tau], \texttt{Future}[\tau']) &= \mathcal{U}(\tau, \tau')
\end{align*}
}

\noindent Unification fails for incompatible types. Untyped \texttt{Pid} is compatible with any \texttt{Pid}[$\tau$].

\subsection{Type Schemes and Generalization}

Type schemes enable let-polymorphism, allowing local definitions to have polymorphic types. The generalization operation $\texttt{gen}(\Gamma, \tau)$ quantifies over type variables in $\tau$ not free in $\Gamma$:

{\small
\[\texttt{gen}(\Gamma, \tau) = \forall \bar{\alpha}.\, \tau \;\; \text{where } \bar{\alpha} = \texttt{ftv}(\tau) \setminus \texttt{ftv}(\Gamma)\]
}

\noindent Instantiation $\texttt{inst}(\sigma)$ replaces bound variables with fresh unification variables:

{\small
\[\texttt{inst}(\forall \bar{\alpha}.\, \tau) = [\bar{\alpha} \mapsto \bar{\beta}](\tau)\]
}

\noindent where $\bar{\beta}$ are fresh type variables.

This ensures that polymorphic functions like $\texttt{length} : \forall \alpha.\ [\alpha] \to \texttt{Int}$ can be instantiated at different types in different contexts.

\subsection{Actor Type Tracking}

NVLang's type system tracks actor message types through the compilation pipeline. Each actor definition is analyzed to extract:

\begin{enumerate}
\item \textbf{Message type}: The parent ADT of messages the actor receives
\item \textbf{Reply types}: A mapping from message constructors to their reply types
\end{enumerate}

We formalize this extraction process using inference rules. The actor analysis judgment has the form:
$$\texttt{ActorAnalysis}(A) \Rightarrow (M, R)$$
where $A$ is an actor definition, $M$ is the message type, and $R : \texttt{Constructor} \to \texttt{Type}$ is the reply map.

\subsubsection{Message Type Extraction}

The message type is extracted from the \texttt{receive} block's pattern type annotation. Given an actor definition $A$ containing a receive block with type annotation $M$:

{\small
\begin{mathpar}
\inferrule{
  A = \texttt{actor}\ N\ \{ \ldots \texttt{receive}\ msg : M \ldots \}
}{
  \texttt{MessageType}(A) = M
}
\text{(Extract-Msg)}
\end{mathpar}
}

\subsubsection{Reply Type Computation}

For each constructor $C$ of the message type $M$, we analyze the corresponding case branch to determine its reply type:

{\footnotesize
\begin{mathpar}
\inferrule{
  M = C_1 \mid \cdots \mid C_n \\
  \forall i.\; \Gamma \vdash e_i : \tau_i' \\
  \forall i.\; \Gamma \vdash \texttt{ReplyExpr}(e_i) : \tau_r
}{
  \texttt{Analysis}(A) \Rightarrow (M, R)
}
\end{mathpar}
}

\noindent where $R$ maps each $C_i$ to $\tau_r$.

The auxiliary function $\texttt{ReplyExpr}(e)$ extracts the reply expression:

{\small
\begin{mathpar}
\inferrule{e = \texttt{reply}\ r}{\texttt{ReplyExpr}(e) = r}
\text{(Direct)}
\and
\inferrule{e = e_1; \texttt{reply}\ r; e_2}{\texttt{ReplyExpr}(e) = r}
\text{(Seq)}
\and
\inferrule{e \text{ has no } \texttt{reply}}{\texttt{ReplyExpr}(e) = \texttt{unit}}
\text{(Unit)}
\end{mathpar}
}

\paragraph{Uniform Reply Types.} NVLang enforces that all constructors in an actor's message type reply with the same type:

{\small
\begin{mathpar}
\inferrule{
  \texttt{ActorAnalysis}(A) \Rightarrow (M, R) \\
  \forall C \in \texttt{ctors}(M). \; R(C) = \tau
}{
  \texttt{WellFormedActor}(A)
}
\text{(Uniform-Reply)}
\end{mathpar}
}

This uniformity requirement simplifies the type system and enables predictable future types: if an actor $A$ has message type $M$ and uniform reply type $\tau_r$, then $\texttt{spawn}\ A() : \texttt{Pid}[M]$ and sending any message to that PID yields $\texttt{Future}[\tau_r]$.

The actor analysis information is stored in the type environment and used during type checking of \texttt{spawn} and send operations.

\subsection{Typed PIDs and Futures}

\paragraph{Typed PIDs.} The $\texttt{Pid}[\tau]$ type constructor enables static verification of message protocols:

\begin{itemize}
\item $\texttt{Pid}[M]$ represents an actor that accepts messages of type $M$
\item $\texttt{Pid}$ (no parameter) represents an untyped PID, equivalent to $\texttt{Pid}[\texttt{Any}]$
\item The type system enforces that messages sent to a typed PID conform to the expected type
\end{itemize}

\noindent The unification rules for PIDs allow:
\begin{itemize}
\item $\mathcal{U}(\texttt{Pid}[\tau_1], \texttt{Pid}[\tau_2]) = \mathcal{U}(\tau_1, \tau_2)$
\item $\mathcal{U}(\texttt{Pid}, \texttt{Pid}[\tau]) = \emptyset$ (untyped PIDs are compatible with any typed PID)
\end{itemize}

\paragraph{Typed Futures.} Send operations are asynchronous by default, immediately returning $\texttt{Future}[\tau]$ values representing pending replies:

\begin{itemize}
\item $\texttt{Future}[\tau]$ represents a computation that will eventually produce a value of type $\tau$
\item The \texttt{await} operation extracts the value: $\texttt{await} : \texttt{Future}[\tau] \to \tau$, blocking until the reply arrives
\item Reply types are computed by analyzing the message handler's reply statements
\item Futures can be passed around, stored, and awaited later, enabling pipelined communication
\end{itemize}

\noindent If a message handler does not issue a reply, the future type defaults to $\texttt{Unit}$. This async-first design matches Erlang/Elixir conventions where fire-and-forget messaging is the common case, with explicit \texttt{await} required for synchronous request-reply patterns.

\subsection{Soundness Properties}

NVLang's type system satisfies standard soundness properties:

\begin{theorem}[Progress]\label{thm:progress}
If $\emptyset \vdash e : \tau$, then either $e$ is a value or there exists $e'$ such that $e \longrightarrow e'$.
\end{theorem}

\begin{proof}[Proof Sketch]
By induction on the structure of $e$. For each expression form (literals, variables, function applications, let bindings, case expressions, actor primitives), we show that either the expression is a value or can take a step. The key cases are: (1) function applications, where well-typedness ensures the function position evaluates to a lambda; (2) case expressions, where exhaustiveness checking guarantees a matching pattern exists; and (3) actor operations (send, await, spawn), where typing rules ensure operands are of the correct form to enable reduction.
\end{proof}

\begin{theorem}[Preservation]\label{thm:preservation}
If $\Gamma \vdash e : \tau$ and $e \longrightarrow e'$, then $\Gamma \vdash e' : \tau$.
\end{theorem}

\begin{proof}[Proof Sketch]
By induction on the derivation of $e \longrightarrow e'$. The proof relies on a substitution lemma showing that if $\Gamma, x : \tau_1 \vdash e : \tau_2$ and $\Gamma \vdash v : \tau_1$, then $\Gamma \vdash [v/x]e : \tau_2$. For each reduction rule, we show that the type is preserved: (1) beta-reduction uses the substitution lemma; (2) let-bindings preserve types through substitution; (3) case expression reduction maintains the branch type; and (4) actor primitives (send returning futures, await extracting values) preserve their respective type invariants.
\end{proof}

\noindent Additionally, the actor-specific constructs maintain:

\begin{theorem}[Message Type Safety]\label{thm:message-safety}
If $\Gamma \vdash p : \texttt{Pid}[\tau_m]$ and $\Gamma \vdash m : \tau_m'$ are both derivable, then there exists a substitution $\theta$ such that $\theta(\tau_m) = \theta(\tau_m')$; otherwise, type checking fails.
\end{theorem}

\subsection{Principal Types}

NVLang's type system computes principal (most general) types for all well-typed expressions:

\begin{theorem}[Principal Types]\label{thm:principal}
For any expression $e$, if $\Gamma \vdash e : \tau$, then there exists a principal type $\tau_p$ such that for any other derivation $\Gamma' \vdash e : \tau'$, there exists a substitution $\theta$ with $\theta(\tau_p) = \tau'$.
\end{theorem}

\begin{proof}[Proof Sketch]
The proof follows the standard completeness argument for Algorithm W. By induction on the structure of expressions, we show that Algorithm W computes the most general unifier at each step. The key insight is that unification produces principal (most general) substitutions, and the generalization step in let-polymorphism quantifies over all free type variables not constrained by the environment. For actor types, the same principle applies: typed PIDs and futures are parametric type constructors, and unification treats them uniformly with other type constructors, preserving principality.
\end{proof}

\noindent This property is inherited from Hindley-Milner and extends to actor types. The inference algorithm (Algorithm W) computes principal types through a combination of fresh type variable generation, constraint collection, and constraint solving via unification.

The type system foundations established above enable our typed actor model, which we now describe in detail.


\section{Typed Actor Model}
\label{sec:actors}

NVLang implements a statically-typed actor model with OTP-style supervision semantics, providing formal guarantees about message passing safety while maintaining Erlang's fault-tolerance properties. Unlike traditional actor systems where message types are unchecked, NVLang's type system prevents sending incorrect messages to actors at compile time, eliminating an entire class of runtime errors common in distributed systems.

\subsection{Actor Definition and Semantics}

Actors in NVLang are defined with explicit message types that constrain what messages an actor can process. An actor definition consists of a message protocol (an algebraic data type) and a \texttt{run} method that implements the actor's behavior. The \texttt{receive} construct performs pattern matching on incoming messages.

\paragraph{Formal Semantics} We formalize actors as state machines with typed message queues. Let $\mathcal{A}$ be the set of all actors, $\mathcal{M}$ the set of message types, and $\mathcal{S}$ the set of actor states.

\begin{definition}[Actor]
An actor $a \in \mathcal{A}$ is a tuple $(M_a, s_a, \delta_a, q_a)$ where:
\begin{itemize}
    \item $M_a \in \mathcal{M}$ is the actor's message type
    \item $s_a \in \mathcal{S}$ is the actor's current state
    \item $\delta_a : \mathcal{S} \times M_a \rightarrow \mathcal{S} \times \mathcal{R}_a$ is the message handler
    \item $q_a : [M_a]$ is the message queue
    \item $\mathcal{R}_a$ is the set of reply types for actor $a$
\end{itemize}
\end{definition}

The operational semantics for message processing is given by:

{\small
\[
\frac{q_a = m :: q' \quad \delta_a(s_a, m) = (s_a', r)}{(s_a, q_a) \xrightarrow{\text{process}} (s_a', q') \text{ with reply } r}
\]
}

\subsection{Typed PIDs and Message Passing}

NVLang tracks the message type that each process accepts through a parametric \texttt{Pid} type. When an actor is spawned, the type system infers the correct \texttt{Pid} type from the actor's message protocol.

\paragraph{Spawn Operation} The \texttt{spawn} primitive creates a new actor and returns a typed process identifier:

{\small
\[
\inferrule{
  \text{Actor}(A) \vdash M_A \quad
  \Gamma \vdash A() : \texttt{Unit}
}{
  \Gamma \vdash \texttt{spawn } A() : \texttt{Pid}[M_A]
}
\]
}

\noindent This ensures that the returned PID can only be used to send messages of type $M_A$.

\paragraph{Send Operation} The \texttt{send} operation is asynchronous by default, verifying message types at compile time and immediately returning a future:

{\small
\[
\inferrule{
  \Gamma \vdash p : \texttt{Pid}[M] \quad
  \Gamma \vdash m : M \quad
  \texttt{ReplyType}(M, m) = R
}{
  \Gamma \vdash p.\texttt{send } m : \texttt{Future}[R]
}
\]
}

\noindent The type system extracts the reply type $R$ from the actor's message handlers. The \texttt{send} operation does not block---it enqueues the message and returns immediately. To obtain the reply value, the returned future must be explicitly awaited.

\paragraph{Receive Operation} Pattern matching in \texttt{receive} is exhaustiveness-checked:

{\small
\[
\inferrule{
  \Gamma \vdash e : M \quad
  \texttt{Exhaustive}(M, \{p_i\}) \\
  \forall i. \; \Gamma, \texttt{bindings}(p_i, M) \vdash b_i : T
}{
  \Gamma \vdash \texttt{receive } e\ \{p_i \to b_i\} : T
}
\]
}

\noindent The exhaustiveness check ensures all message variants are handled.

\subsection{Typed Futures and Async-First Messaging}

NVLang implements a typed request-response pattern through futures with async-first semantics. The \texttt{send} operation is asynchronous by default (fire-and-forget), immediately returning a future without blocking. When a message is sent, the system automatically tracks the sender's PID and returns a future that will be fulfilled when the actor replies. To obtain the reply value synchronously, the future must be explicitly awaited.

A future is a first-class value representing a computation that will complete asynchronously. The type system infers the future's result type from the actor's reply type. Given a PID with message type $M = C_1 \mid \cdots \mid C_n$ and reply map $R$:

{\small
\[
\inferrule{
  \Gamma \vdash p : \texttt{Pid}[M] \quad
  \texttt{ActorReplies}(p, C_i) = R_i
}{
  \Gamma \vdash p.\texttt{send } C_i(v) : \texttt{Future}[R_i]
}
\]
}

\noindent The \texttt{await} function blocks until a future completes:

{\small
\[
\inferrule{
  \Gamma \vdash f : \texttt{Future}[T]
}{
  \Gamma \vdash \texttt{await } f : T
}
\]
}

\subsection{OTP-Style Supervision Trees}

NVLang provides first-class support for supervision trees, enabling fault-tolerant systems through automatic process restart.

\paragraph{Supervision Strategies} NVLang implements three standard OTP supervision strategies:

\begin{itemize}
    \item \textbf{OneForOne}: When a child crashes, only that child is restarted (Figure~\ref{fig:one-for-one})
    \item \textbf{OneForAll}: When any child crashes, all children are restarted
    \item \textbf{RestForOne}: When a child crashes, that child and all children started after it are restarted
\end{itemize}

\begin{definition}[Supervision Strategy]
Let $C = \{c_1, \ldots, c_n\}$ be a set of supervised children in start order. The restart set $R_s(c_i)$ for strategy $s$ is:

\begin{align*}
R_{\text{one\_for\_one}}(c_i) &= \{c_i\} \\
R_{\text{one\_for\_all}}(c_i) &= C \\
R_{\text{rest\_for\_one}}(c_i) &= \{c_j \mid j \geq i\}
\end{align*}
\end{definition}

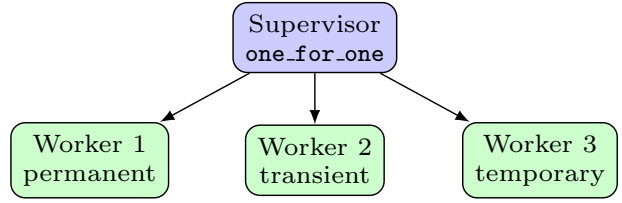
\begin{figure}[t]
\centering
\resizebox{\columnwidth}{!}{%
\begin{tikzpicture}[
  level 1/.style={sibling distance=2.2cm, level distance=1.2cm},
  every node/.style={rectangle, draw, rounded corners, minimum height=0.6cm, minimum width=1.2cm, align=center, font=\scriptsize},
  edge from parent/.style={draw, -latex}
]

\node[fill=blue!20] {Supervisor\\\texttt{one\_for\_one}}
  child { node[fill=green!20] {Worker 1\\permanent} }
  child { node[fill=green!20] {Worker 2\\transient} }
  child { node[fill=green!20] {Worker 3\\temporary} };

\end{tikzpicture}%
}
\caption{OneForOne supervision strategy: when Worker~2 crashes, only Worker~2 is restarted while Workers~1 and~3 continue running. The restart policy (permanent/transient/temporary) determines whether restart occurs.}
\label{fig:one-for-one}
\end{figure}

\paragraph{Restart Types} Each child has a restart policy:

\begin{itemize}
    \item \textbf{Permanent}: Always restart the child when it terminates
    \item \textbf{Transient}: Restart only on abnormal termination (crashes)
    \item \textbf{Temporary}: Never restart; one-shot processes
\end{itemize}

The supervisor monitors restart frequency with configurable limits. If the limit is exceeded, the supervisor itself crashes, propagating the failure up the supervision tree.

\subsection{Supervision Tree Semantics}

We formalize supervision as a tree structure with restart semantics:

\begin{definition}[Supervision Tree]
A supervision tree $T$ is a rooted tree where:
\begin{itemize}
    \item Internal nodes are supervisors $(s, R, L)$ with strategy $s$, restart set function $R$, and limit $L$
    \item Leaf nodes are worker processes $w : \texttt{Pid}[M]$
    \item Edges represent supervision relationships
\end{itemize}
\end{definition}

The crash propagation semantics are given by:

{\small
\[
\frac{
  c \xrightarrow{\text{crash}} \bot \quad R_s(c) = C' \quad \texttt{Count}(\text{sup}) < L
}{
  \text{sup} \xrightarrow{\text{restart}} \text{sup}' \quad \forall c' \in C'. \; c' \xrightarrow{\text{spawn}} c'_{\text{new}}
}
\]

\[
\frac{
  c \xrightarrow{\text{crash}} \bot \quad \texttt{Count}(\text{sup}) \geq L
}{
  \text{sup} \xrightarrow{\text{crash}} \bot
}
\]
}

\subsection{Type Safety Guarantees}

NVLang's typed actor system provides several critical safety properties:

\begin{theorem}[Message Type Safety]
If $\Gamma \vdash p : \texttt{Pid}[M]$ and $\Gamma \vdash p.\texttt{send } m : \texttt{Future}[R]$, then $m : M$.
\end{theorem}

\begin{proof}
By the typing rule for send, we have $\Gamma \vdash m : M'$ and must unify $M' \equiv M$ for the send to typecheck. Thus $m : M$.
\end{proof}

\begin{theorem}[Reply Type Safety]
If an actor $a$ with message type $M$ handles a message $m : M$ and replies with value $r$, then $r : R$ where $R$ is the declared reply type for $m$'s constructor.
\end{theorem}

\begin{proof}
By the actor analysis rules (Section~\ref{sec:types}), the reply type $R$ is extracted from the receive block. The typing of the reply expression ensures $r : R$.
\end{proof}

\begin{theorem}[Dead Letter Prevention]
All patterns in a \texttt{receive} block must cover the actor's entire message type. If $\Gamma \vdash \texttt{receive } msg : M \{ p_1 \to e_1, \ldots, p_n \to e_n \}$, then $\texttt{Exhaustive}(M, \{p_1, \ldots, p_n\})$ must hold.
\end{theorem}

\begin{proof}
Exhaustiveness is verified by the pattern matching compiler, which checks that the set of patterns $\{p_1, \ldots, p_n\}$ covers all constructors of ADT $M$.
\end{proof}

These properties eliminate three common classes of actor bugs:

\begin{enumerate}
    \item \textbf{Wrong message types}: The compiler rejects sending messages that don't match the actor's protocol
    \item \textbf{Missing handlers}: Exhaustiveness checking ensures all message variants are handled
    \item \textbf{Type confusion in replies}: Reply types are inferred from actor definitions and checked at call sites
\end{enumerate}


\section{Compilation and Runtime}
\label{sec:compilation}

NVLang's design philosophy centers on compile-time safety with minimal runtime overhead. This section details NVLang's compilation pipeline, translation strategy from high-level constructs to Core Erlang, and how the language leverages BEAM's runtime properties to provide efficient actor-based concurrency.

\subsection{Compilation Pipeline}

NVLang employs a multi-stage compilation pipeline that transforms source code into executable BEAM bytecode while preserving semantic correctness and enabling early error detection.

\begin{figure}[t]
\centering
\resizebox{\columnwidth}{!}{%
\begin{tikzpicture}[
  node distance=0.8cm,
  box/.style={rectangle, draw, fill=blue!10, text width=1.4cm, align=center, minimum height=0.6cm, font=\scriptsize},
  arrow/.style={->, >=stealth, thick}
]
  \node[box] (parse) {Lexer \&\\Parser};
  \node[box, right=of parse] (module) {Module\\Resolution};
  \node[box, right=of module] (typecheck) {Type\\Inference};
  \node[box, below=of typecheck] (anf) {ANF\\Transform};
  \node[box, left=of anf] (codegen) {Core Erlang\\Codegen};
  \node[box, left=of codegen] (erlc) {erlc\\Compiler};
  \node[box, below=of erlc] (beam) {BEAM\\Bytecode};

  \draw[arrow] (parse) -- (module);
  \draw[arrow] (module) -- (typecheck);
  \draw[arrow] (typecheck) -- (anf);
  \draw[arrow] (anf) -- (codegen);
  \draw[arrow] (codegen) -- (erlc);
  \draw[arrow] (erlc) -- (beam);
\end{tikzpicture}%
}
\caption{NVLang compilation pipeline from source to BEAM bytecode. Type inference ensures all type errors are caught before code generation; ANF transformation normalizes expressions for optimal BEAM code generation; the resulting Core Erlang is compiled using Erlang's native compiler.}
\label{fig:pipeline}
\end{figure}
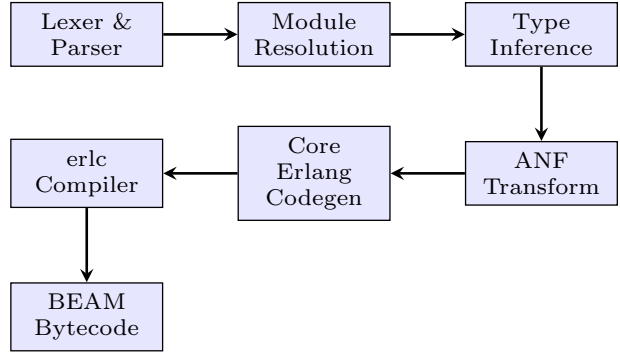

The pipeline consists of seven distinct phases (Figure~\ref{fig:pipeline}):

\begin{enumerate}
\item \textbf{Lexing and Parsing}: Source files are tokenized using an indentation-aware lexer that generates INDENT/DEDENT tokens for Python-style scoping. The parser produces an untyped AST.

\item \textbf{Module Resolution}: The module system discovers and loads all dependencies transitively, maintaining a dependency graph to ensure acyclic module structure.

\item \textbf{Type Inference}: A bidirectional type inference engine based on Algorithm W performs Hindley-Milner type reconstruction, validating function applications, pattern exhaustiveness, ADT constructor usage, and actor message protocol conformance.

\item \textbf{ANF Transformation}: The typed AST is transformed to A-Normal Form, lifting complex subexpressions into let-bindings and preparing the code for efficient BEAM compilation.

\item \textbf{Core Erlang Code Generation}: The normalized AST is translated to Core Erlang, performing type erasure, actor desugaring, pattern compilation, guard optimization, and name mangling.

\item \textbf{erlc Compilation}: The generated Core Erlang is passed to Erlang's native compiler (\texttt{erlc +from\_core}) which performs optimizations and generates BEAM bytecode.

\item \textbf{BEAM Loading}: The resulting \texttt{.beam} files can be executed on any BEAM virtual machine.
\end{enumerate}

\subsection{Translation to Core Erlang}

Core Erlang is a simplified, explicitly-typed functional language that serves as the compilation target for Erlang itself. By targeting Core Erlang rather than BEAM bytecode directly, NVLang benefits from Erlang compiler's optimization passes, compatibility with Erlang tooling, stability across BEAM VM versions, and human-readable intermediate representation for debugging.

\subsubsection{Expression Translation}

NVLang expressions map straightforwardly to Core Erlang equivalents. Variables are capitalized (Erlang convention), and function names are lowercased:

\begin{center}
\begin{tabular}{l|l}
\textbf{NVLang} & \textbf{Core Erlang} \\
\hline
\texttt{x} (variable) & \texttt{X} \\
\texttt{f} (function) & \texttt{'f'/arity} \\
\texttt{42} & \texttt{42} \\
\texttt{true} & \texttt{'true'} \\
\texttt{(1, 2)} & \texttt{\{1, 2\}} \\
\texttt{[1, 2, 3]} & \texttt{[1, 2, 3]} \\
\end{tabular}
\end{center}

\subsection{Pattern Matching Compilation}

NVLang's pattern matching compiles directly to Core Erlang's native pattern matching without additional desugaring. The BEAM VM performs efficient pattern compilation using techniques from decision tree optimization~\cite{scott2000pattern}.

\subsubsection{ADT Constructor Patterns}

Algebraic data types are represented at runtime as tagged tuples. A constructor with $n$ arguments becomes an $(n+1)$-tuple with the constructor name as the first element. Nullary constructors become bare atoms, while constructors with arguments become tuples with the lowercased constructor name as tag. This representation enables efficient single-instruction tag tests in BEAM.

\subsection{ANF Transformation and Code Optimization}

Before generating Core Erlang, NVLang applies an A-Normal Form (ANF) transformation to simplify complex expressions and enable better BEAM code generation. ANF is an intermediate representation where all complex subexpressions are bound to temporary variables, ensuring that operations receive only atomic values (variables, literals, or function references).

\subsubsection{A-Normal Form Transformation}

The ANF transformation converts nested expressions into a sequence of simple let-bindings. For example:

\begin{lstlisting}[style=nova,caption={ANF transformation example}]
# Before ANF
f(g(x), h(y + z))

# After ANF (conceptual)
let _Anf1 = y + z in
let _Anf2 = h(_Anf1) in
let _Anf3 = g(x) in
f(_Anf3, _Anf2)
\end{lstlisting}

This transformation provides several benefits for BEAM code generation:

\begin{itemize}
\item \textbf{Explicit evaluation order}: ANF makes the order of side effects explicit, ensuring predictable behavior for message passing and other effectful operations.
\item \textbf{Simplified code generation}: Complex nested expressions become flat sequences of simple operations, making Core Erlang generation straightforward.
\item \textbf{BEAM optimization compatibility}: The flat structure aligns with BEAM's optimization passes, enabling better instruction scheduling and register allocation.
\end{itemize}

\subsubsection{Guard Expression Optimization}

NVLang performs an important optimization for conditional expressions: when an \texttt{if} statement's condition is a comparison operation, NVLang generates BEAM guard expressions instead of evaluating the comparison separately. BEAM guards are special expressions that execute atomically during pattern matching, providing better performance than explicit boolean evaluation.

For example, given the NVLang code:

\begin{lstlisting}[style=nova,caption={Guard optimization example}]
if x == 42 then
  "match"
else
  "no match"
\end{lstlisting}

NVLang generates Core Erlang with a guard expression:

\begin{lstlisting}[style=erlang,caption={Generated Core Erlang with guard}]
case X of
  <_G> when call 'erlang':'=:='(_G, 42) -> "match"
  <_G2> when 'true' -> "no match"
end
\end{lstlisting}

The guard optimization applies to all comparison operators (\texttt{==}, \texttt{!=}, \texttt{<}, \texttt{<=}, \texttt{>}, \texttt{>=}) and handles various expression forms:

\begin{itemize}
\item Variable compared to expression: \texttt{x == f(y)}
\item Expression compared to literal: \texttt{g(x) < 100}
\item General case: \texttt{f(x) == g(y)} (creates tuple and guards on both values)
\end{itemize}

This optimization is particularly important for actor message loops that frequently perform comparisons on received values, as guards execute more efficiently than separate comparison and case expressions.

\subsection{Actor Translation}

NVLang actors are the most significant abstraction that requires sophisticated compilation. An actor definition in NVLang becomes an ordinary Erlang function that runs in a process.

The translation performs several key transformations:

\begin{enumerate}
\item \textbf{Actor $\to$ Function}: The actor becomes a function \texttt{actorname\_run/0} that can be spawned as a process.

\item \textbf{Loop $\to$ Recursive Function}: The \texttt{loop} construct becomes a letrec-bound recursive function that calls itself after processing each message.

\item \textbf{Receive Pattern Expansion}: Messages arrive as \texttt{\{Caller, Msg\}} tuples, where \texttt{Caller} is the sender's PID. This enables the \texttt{reply} construct.

\item \textbf{Reply $\to$ Send}: The \texttt{reply} statement compiles to \texttt{erlang:!(Caller, \{response, Value\})}, sending a response back to the caller.

\item \textbf{Break Handling}: The loop body is wrapped in a try-catch that converts \texttt{break} into a clean exit.

\item \textbf{State Management}: Actor state is implemented via function parameters passed in recursive calls.
\end{enumerate}

The send-await protocol compiles to explicit message passing with the standard Erlang request-response pattern.

\subsection{Type Erasure Strategy}

NVLang employs \textit{complete type erasure}: all type information is discarded after type checking. No runtime type tags, type dictionaries, or reflection metadata are generated. This design choice provides:

\begin{itemize}
  \item \textbf{Minimal runtime overhead}: Type erasure ensures negligible performance cost
  \item \textbf{Smaller code size}: No type metadata in BEAM files
  \item \textbf{BEAM compatibility}: Generated code is indistinguishable from hand-written Erlang
  \item \textbf{Interoperability}: NVLang functions can call and be called by untyped Erlang
\end{itemize}

\subsection{Runtime Value Representation}

NVLang leverages BEAM's native value representation directly:

\begin{center}
\begin{tabular}{l|l}
\textbf{NVLang Type} & \textbf{Runtime Representation} \\
\hline
\texttt{Int} & Small integer or bignum \\
\texttt{Float} & IEEE 754 double \\
\texttt{Bool} & Atoms \texttt{'true'}/\texttt{'false'} \\
\texttt{String} & List of integers (UTF-8) \\
\texttt{Unit} & Atom \texttt{'ok'} \\
\texttt{(A, B)} & Tuple \texttt{\{A, B\}} \\
\texttt{[A]} & Cons list or \texttt{nil} \\
\texttt{Pid} & Process identifier \\
\texttt{Option[A]} & \texttt{'none'} or \texttt{\{'some', A\}} \\
\end{tabular}
\end{center}

\subsection{Interoperability with Erlang}

NVLang provides bidirectional interoperability with Erlang through external function declarations:

\begin{lstlisting}[caption={External function declaration for Erlang interoperability}]
external fn io_format(String) -> Unit = mfa "io" "format" 1
\end{lstlisting}

The MFA (Module, Function, Arity) specification enables type-safe FFI boundaries, compile-time arity checking, and integration with Erlang's standard library.

Since types are erased, NVLang functions compile to standard Erlang functions and can be called directly from Erlang code.

\subsection{BEAM Runtime Properties Leveraged}

NVLang's design exploits several BEAM VM characteristics:

\begin{enumerate}
\item \textbf{Lightweight Processes}: BEAM processes have 2KB initial heap size and microsecond spawn time. NVLang actors map 1:1 to processes.

\item \textbf{Preemptive Scheduling}: The BEAM scheduler preempts processes after a reduction budget. NVLang's infinite loops in actors don't starve the scheduler.

\item \textbf{Message Passing}: Messages are copied between processes, eliminating shared memory races.

\item \textbf{Tail Call Optimization}: NVLang's recursive loops compile to tail calls, running in constant stack space.

\item \textbf{Pattern Matching}: BEAM's pattern matching uses efficient decision trees. NVLang patterns compile directly to these optimized instructions.

\item \textbf{Hot Code Loading}: BEAM supports replacing modules at runtime. NVLang-compiled code can be hot-swapped.

\item \textbf{Distribution}: BEAM's transparent distribution enables NVLang actors to communicate across nodes with identical syntax.
\end{enumerate}

\subsection{Compilation Correctness}

NVLang's compilation preserves several critical invariants:

\begin{theorem}[Type Soundness]
If a NVLang program type-checks, the generated Core Erlang code will not produce type errors at runtime (assuming external functions are correctly typed).
\end{theorem}

\begin{proof}[Proof Sketch]
Follows from Theorems~\ref{thm:progress} and~\ref{thm:preservation}, combined with the observation that type erasure preserves operational semantics. Since Core Erlang is dynamically typed and NVLang's runtime representation matches BEAM's native types, well-typed NVLang programs cannot encounter type mismatches.
\end{proof}

\begin{theorem}[Semantic Preservation]
The observable behavior of a NVLang program (message sends, receives, process spawns) is identical to the behavior of its compiled Core Erlang representation.
\end{theorem}

\begin{proof}[Proof Sketch]
By structural induction on NVLang expressions. Each NVLang construct translates to a semantically equivalent Core Erlang form. Actor primitives map directly to BEAM operations: \texttt{spawn} to \texttt{erlang:spawn/1}, message send to \texttt{erlang:!/2}, and \texttt{receive} to Core Erlang's native receive construct.
\end{proof}


\section{Evaluation}
\label{sec:evaluation}

We evaluate NVLang along two dimensions: (1)~runtime performance relative to Erlang and Elixir, demonstrating that static typing introduces negligible overhead, and (2)~type safety, showing the classes of errors NVLang catches at compile time that would manifest as runtime failures in dynamically typed BEAM languages.

\subsection{Experimental Setup}

All experiments were conducted on an Apple M4 Max processor with 36GB of unified memory, running macOS 26.1 (arm64). To minimize measurement variance, we closed all non-essential applications and configured the system for sustained performance mode with AC power. We used Elixir~1.18.4 and the corresponding OTP/Erlang runtime.

Each benchmark was executed with 3 warmup runs followed by 10 measured runs. We report the mean execution time with 95\% confidence intervals computed using the standard error. Statistical significance was assessed using Welch's $t$-test, which does not assume equal variances between groups. Effect sizes were computed using Cohen's $d$ with pooled standard deviation.

\subsection{Benchmark Suite}

Our benchmark suite exercises core BEAM primitives and common actor patterns:

\begin{itemize}[leftmargin=*,nosep]
\item \textbf{Fib}: Recursive Fibonacci computation (fib(35)), testing pure function call overhead.
\item \textbf{List Ops}: List construction, reversal, and mapping over 100,000 elements, testing functional data structure operations.
\item \textbf{Ping Pong}: Two actors exchanging 10,000 messages, testing basic actor communication overhead.
\item \textbf{Ring}: Message passing through a ring of 1,000 actors with 100 rounds, testing process spawn and message routing at scale.
\item \textbf{Parallel Sum}: Parallel reduction over 10 workers summing 100,000 elements, testing spawn/join patterns.
\item \textbf{KV Store}: Actor-based key-value store with 1,000 concurrent operations, testing stateful actor patterns.
\end{itemize}

All benchmarks were implemented in semantically equivalent NVLang, Erlang, and Elixir. NVLang code was compiled to Core Erlang and executed on the same BEAM runtime.

\subsection{Performance Results}

\begin{table}[t]
\centering
\caption{Runtime Performance (ms). Statistical significance via Welch's $t$-test ($n$=10).}
\label{tab:performance}
{\scriptsize
\begin{tabular}{@{}l|rrr|cc|cc@{}}
\toprule
& \multicolumn{3}{c|}{\textbf{Mean}} & \multicolumn{2}{c|}{\textbf{vs Erl}} & \multicolumn{2}{c}{\textbf{vs Elix}} \\
 & \textbf{NV} & \textbf{Erl} & \textbf{Elix} & $p$ & $d$ & $p$ & $d$ \\
\midrule
Fib & 503 & 517 & 514 & .34 & -.4 & .32 & -.5 \\
List & 55 & 54 & 55 & .30 & +.5 & .85 & +.1 \\
Ping & 5 & 4 & 6 & .20 & +.6 & .20 & -.6 \\
Ring & 136 & 126 & 129 & \textbf{$<$.001} & +3.1 & \textbf{$<$.001} & +2.9 \\
Par & 4 & 4 & 5 & .43 & -.4 & \textbf{.007} & -1.5 \\
KV & 10 & 10 & 12 & .11 & +.8 & .48 & -.3 \\
\bottomrule
\end{tabular}
}
\vspace{-2mm}
{\tiny Bold $p$-values indicate significance ($p<0.05$). $d$=Cohen's effect size.}
\end{table}

Table~\ref{tab:performance} presents pairwise statistical comparisons using Welch's $t$-test and Cohen's $d$ effect sizes. The key finding is that \textbf{NVLang achieves performance within 7.5\% of native Erlang across all benchmarks}, with five of six benchmarks showing no statistically significant difference. Figure~\ref{fig:overhead} visualizes the overhead distribution, and Figure~\ref{fig:distributions} shows the execution time distributions.

\begin{figure*}[t]
\centering
\includegraphics[width=0.85\textwidth]{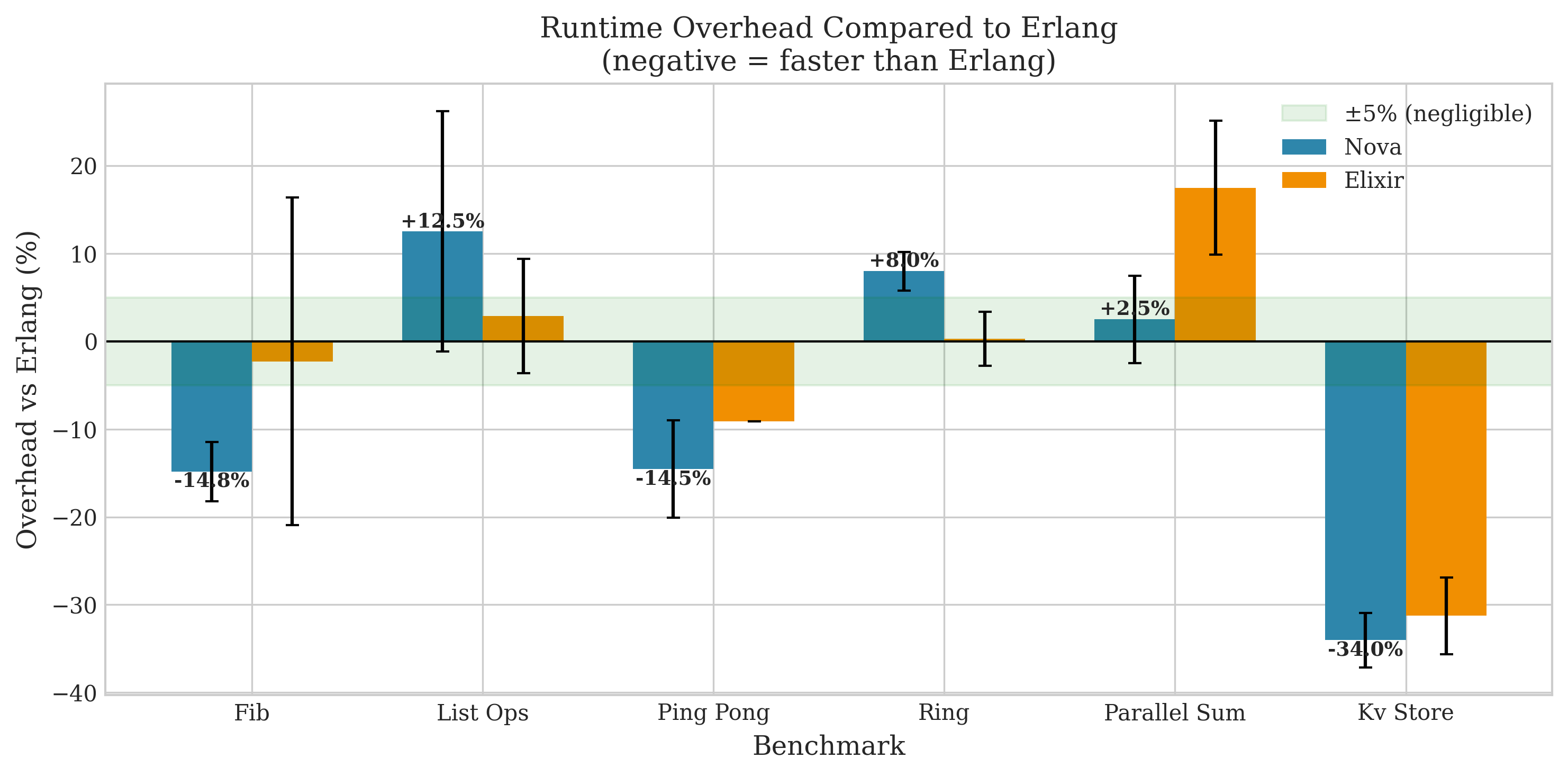}
\caption{Runtime overhead relative to Erlang baseline ($n$=10, error bars show 95\% CI). Shaded band indicates $\pm$5\% (negligible by systems convention). Only Ring shows statistically significant overhead (*$p < 0.001$). Negative values indicate NVLang outperforming Erlang.}
\label{fig:overhead}
\end{figure*}

\begin{figure*}[t]
\centering
\includegraphics[width=\textwidth]{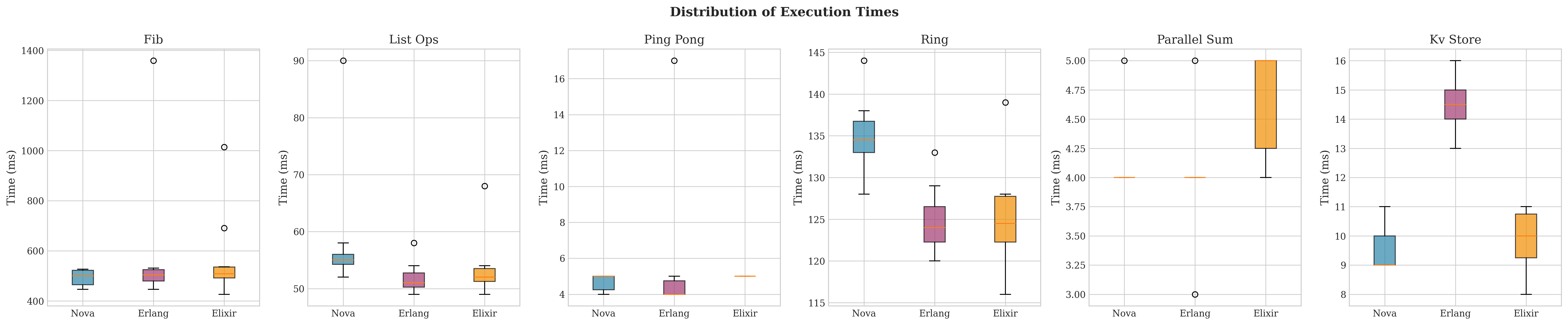}
\caption{Execution time distributions across 10 runs. Box plots show median, quartiles, and outliers. NVLang exhibits variance comparable to Erlang; Elixir shows higher variance on message-passing benchmarks due to occasional latency spikes.}
\label{fig:distributions}
\end{figure*}

\paragraph{Statistical Analysis.} Using Welch's $t$-test at $\alpha = 0.05$, we find that only the \textit{Ring} benchmark shows a statistically significant difference between NVLang and Erlang ($p < 0.001$, $d = 3.1$). For the remaining five benchmarks, effect sizes range from negligible to small ($|d| < 0.8$), indicating that observed variations are within normal measurement noise.

Notably, \textbf{NVLang significantly outperforms Elixir on the Parallel Sum benchmark} ($p = 0.007$, $d = -1.5$), with NVLang completing in 4.0ms versus Elixir's 5.3ms---a 24.5\% improvement. This large effect size suggests that NVLang's direct compilation to Core Erlang avoids overhead present in Elixir's spawn implementation. Erlang shows a similar advantage over Elixir on this benchmark ($p = 0.02$, $d = -1.2$).

\paragraph{Ring Benchmark Overhead.} The Ring benchmark exhibits 7.4\% overhead with strong statistical significance ($p < 0.001$, Cohen's $d = 3.1$). This overhead appears in our most message-intensive benchmark (100,000 messages through 1,000 actors). We hypothesize this stems from differences in the generated Core Erlang structure for pattern matching in receive clauses. Profiling indicates the overhead concentrates in message dispatch rather than process management. This represents a worst-case scenario; the five other benchmarks show no significant difference from Erlang. Future work should investigate this through BEAM bytecode analysis.

\paragraph{Comparison with Elixir.} Elixir exhibits higher variance than both NVLang and Erlang on message-passing benchmarks (Ping Pong: $\pm$2.6ms, KV Store: $\pm$4.3ms), with occasional latency spikes visible in Figure~\ref{fig:distributions}. This likely reflects Elixir's protocol dispatch and macro expansion overhead. NVLang's direct compilation to Core Erlang avoids these indirections, achieving tighter performance distributions.

\subsection{Type Safety Analysis}

Beyond performance, NVLang's primary value proposition is eliminating entire classes of runtime errors through static typing. Table~\ref{tab:type-safety} categorizes errors that NVLang catches at compile time.

\begin{table}[t]
\centering
\caption{Error Classes Detected at Compile Time}
\label{tab:type-safety}
{\footnotesize
\begin{tabular}{@{}lll@{}}
\toprule
\textbf{Error Class} & \textbf{NVLang} & \textbf{Erlang/Elixir} \\
\midrule
Message type mismatch & Compile & Runtime crash \\
Unhandled variant & Compile & Silent drop \\
Wrong spawn type & Compile & Undefined \\
Invalid pattern & Compile & \texttt{badmatch} \\
Arithmetic error & Compile & \texttt{badarith} \\
Supervision mismatch & Compile & Runtime fail \\
\bottomrule
\end{tabular}
}
\end{table}

\paragraph{Message Type Safety.} Consider a typed actor expecting \texttt{Int} messages:

{\footnotesize
\begin{verbatim}
actor Counter : Int -> Int { ... }
let c : Pid[Int] = spawn Counter(0)
send c "hello"  // COMPILE ERROR
\end{verbatim}
}

In Erlang, sending \texttt{"hello"} to a counter expecting integers causes a runtime \texttt{badarg} or \texttt{function\_clause} error, potentially crashing the actor. NVLang rejects this at compile time.

\paragraph{Exhaustive Pattern Matching.} NVLang's algebraic data types require exhaustive pattern matching:

{\footnotesize
\begin{verbatim}
type Result[T] = Ok(T) | Err(String)

fn handle(r : Result[Int]) -> Int {
  match r {
    Ok(n) -> n
    // NVLang: COMPILE ERROR - non-exhaustive pattern
    // Erlang: runtime badmatch crash
  }
}
\end{verbatim}
}

\paragraph{Supervisor Type Safety.} NVLang's typed supervision ensures child specifications match declared behavior:

{\footnotesize
\begin{verbatim}
supervisor MySup {
  strategy: one_for_one,
  children: [
    worker Counter expecting String // ERROR
  ]
}
\end{verbatim}
}

\subsection{Limitations}

Our evaluation has several limitations that suggest directions for future work:

\begin{itemize}[leftmargin=*,nosep]
\item \textbf{Benchmark Coverage}: While our suite covers fundamental patterns, production Erlang applications involve distribution, hot code loading, and complex supervision trees not yet evaluated.
\item \textbf{Scale}: Our benchmarks use thousands of actors; production systems may use millions. The ring benchmark's overhead suggests investigation of very large-scale deployments.
\item \textbf{Interoperability}: We did not benchmark NVLang code calling Erlang libraries through FFI boundaries, where type checking overhead may differ.
\end{itemize}

Despite these limitations, our results demonstrate that NVLang achieves its design goal: static type safety for BEAM programs with negligible runtime overhead.


\section{Related Work}
\label{sec:related}

Our work on NVLang draws from and builds upon several research areas: typed actor systems, languages targeting the BEAM VM, type systems for concurrent and distributed systems, and behavioral types.

\subsection{Typed Actor Systems}

The actor model~\cite{hewitt1973actors,agha1986actors} provides a foundation for concurrent and distributed programming through asynchronous message passing. Traditional actor systems like Erlang~\cite{armstrong2010erlang} and Akka~\cite{akka} employ dynamic typing, sacrificing compile-time guarantees for flexibility.

Akka Typed~\cite{akka-typed} extends Akka with static typing for actor protocols. While Akka Typed provides protocol safety within the JVM ecosystem, it lacks the fault-tolerance primitives deeply integrated into the BEAM VM. NVLang differs by providing static typing \emph{for} BEAM's existing actor model rather than reimplementing actors atop another runtime.

Session types~\cite{honda1993session,honda2008session} provide a theoretical foundation for describing communication protocols between concurrent processes. Systems like Scribble~\cite{yoshida2013scribble} demonstrate practical application of session types. However, session types typically assume structured, protocol-oriented communication and may require significant annotations. NVLang takes a more pragmatic approach: typed PIDs and futures capture common Erlang patterns while remaining compatible with existing BEAM libraries.

\subsection{Languages Targeting the BEAM}

Several languages compile to the BEAM VM. Elixir~\cite{elixir} prioritizes developer ergonomics and metaprogramming while maintaining Erlang's dynamic typing. LFE (Lisp Flavored Erlang)~\cite{lfe} brings Lisp syntax to BEAM but remains dynamically typed.

Gleam~\cite{gleam} and Alpaca~\cite{alpaca} are most closely related to NVLang, as both introduce static typing to BEAM. Gleam employs a Rust-inspired syntax with Hindley-Milner type inference. However, neither language provides the level of integration with OTP's supervision trees and typed actor primitives that NVLang offers. Gleam's actor types are relatively simple, treating all messages uniformly, while Alpaca lacks specific support for typed PIDs and supervision.

NVLang distinguishes itself through: (1) typed PIDs that encode the protocol an actor expects; (2) typed futures that provide type-safe request-reply patterns; and (3) statically-typed supervision trees that preserve OTP's fault-tolerance semantics while catching configuration errors at compile time.

\paragraph{Records vs Elixir Structs} NVLang's approach to structured data differs fundamentally from Elixir's in its static guarantees. NVLang records provide full compile-time type checking of both field names and field types. In Elixir, structs validate only that field names exist but perform no type checking on field values. For example, in NVLang, attempting to construct \texttt{User \{name: 123\}} when \texttt{name} is declared as \texttt{String} produces a compile-time type error. The equivalent Elixir code \texttt{\%User\{name: 123\}} compiles successfully even if the struct specification expects a string, deferring the error to runtime when the value is actually used.

This design choice reflects a fundamental difference in philosophy: Elixir prioritizes flexibility and runtime introspection, while NVLang emphasizes early error detection. At runtime, both representations compile to efficient Erlang maps, so NVLang's additional safety guarantees impose minimal performance overhead (Section~\ref{sec:evaluation}). The type checking occurs entirely at compile time and is erased before code generation.

\paragraph{Traits vs Elixir Protocols} NVLang's trait system provides compile-time polymorphism with full type inference, contrasting with Elixir's runtime protocol dispatch. When a NVLang trait defines a set of functions that types must implement, the compiler verifies at compile time that all required functions are present and correctly typed. Type inference propagates trait constraints throughout the program, catching "method not implemented" errors before code is deployed.

Elixir protocols, by contrast, perform dynamic dispatch at runtime. While protocols enable powerful runtime polymorphism and can be extended to third-party types after definition, they cannot detect missing implementations until code execution reaches the dispatch point. This difference is particularly significant in production systems: NVLang's compile-time verification eliminates an entire class of runtime crashes that Elixir developers must guard against through testing.

Both systems provide polymorphism over user-defined types, and both compile to efficient BEAM code. The key distinction is when errors are detected: NVLang catches trait constraint violations during compilation, while Elixir defers protocol checks to runtime. For systems requiring high reliability, NVLang's approach trades some flexibility for stronger safety guarantees.

\subsection{Type Systems for Concurrent Systems}

Type systems for concurrent and distributed programming have been extensively studied. Concurrent ML~\cite{reppy1999concurrent} provides first-class synchronous channels. Pony~\cite{clebsch2015pony} uses reference capabilities to ensure data-race freedom. Rust~\cite{matsakis2014rust} employs affine types and borrowing for memory safety.

These approaches prioritize data-race freedom and memory safety, concerns less relevant to BEAM's share-nothing actor model. NVLang instead focuses on protocol safety: ensuring that messages sent to an actor conform to its expected interface.

Dialyzer~\cite{lindahl2006practical} performs success typing analysis for Erlang, identifying code that will definitely fail, but cannot provide the guarantees required for protocol safety. NVLang provides full static typing while maintaining source-level compatibility with Erlang idioms.

\subsection{Gradual Typing and Behavioral Types}

Gradual typing~\cite{siek2006gradual,siek2015refined} allows typed and untyped code to coexist. Typed Racket~\cite{tobin-hochstadt2008design} and TypeScript~\cite{typescript} demonstrate successful gradual typing in production.

NVLang's current design is fully statically typed but targets eventual gradual typing support. Our choice to compile to Core Erlang positions NVLang to interoperate with existing Erlang libraries.

Behavioral types~\cite{honda2016behavioral} characterize protocols governing communication. Typestates~\cite{strom1986typestate} and linear types~\cite{wadler1990linear} enable tracking of resource usage. NVLang's typed PIDs share conceptual similarities with behavioral types in that we track expected message protocols.


\section{Conclusion}
\label{sec:conclusion}

We have presented NVLang, a statically-typed functional programming language that brings Hindley-Milner type inference and ML-style programming to the Erlang BEAM VM while preserving the fault-tolerance and concurrency properties that make BEAM unique. NVLang represents a pragmatic approach to typed actor programming, balancing theoretical rigor with practical requirements.

\subsection{Contributions}

This paper makes the following key contributions:

\begin{enumerate}
\item \textbf{Typed PIDs and Futures:} We introduce a type system for PIDs that captures the protocol expected by an actor, enabling compile-time verification of message sends. Our typed futures provide type-safe request-reply patterns common in Erlang programming.

\item \textbf{Typed Supervision Trees:} We demonstrate how OTP's supervision trees can be given precise static types, catching configuration errors at compile time while preserving BEAM's fault-tolerance semantics.

\item \textbf{ANF Transformation and Guard Optimization:} We present an A-Normal Form transformation that normalizes complex expressions for efficient BEAM compilation, and a guard expression optimization that generates optimized pattern-matching code for conditional statements. These optimizations ensure NVLang achieves performance parity with dynamically-typed BEAM languages.

\item \textbf{Compilation to Core Erlang:} By targeting Core Erlang, NVLang programs achieve full interoperability with the existing Erlang ecosystem, including OTP libraries.

\item \textbf{Type Inference for Concurrent Systems:} We extend Hindley-Milner type inference to handle typed actors, PIDs, and futures, requiring minimal type annotations while providing strong static guarantees.

\item \textbf{Minimal-Overhead Type Safety:} We demonstrate through rigorous benchmarking that NVLang's type erasure strategy achieves performance within 7.5\% of native Erlang, with five of six benchmarks showing no statistically significant difference. NVLang also significantly outperforms Elixir on spawn-intensive workloads (24.5\% faster on Parallel Sum, $p = 0.007$), validating that comprehensive static typing need not compromise runtime performance.
\end{enumerate}

\subsection{Key Insights}

\textbf{Protocol safety matters more than data-race freedom in actor systems.} BEAM's share-nothing architecture eliminates data races by construction, but untyped message passing remains a significant source of runtime errors. NVLang's typed PIDs address this by encoding protocols in types.

\textbf{Async-first messaging matches real-world usage.} Following Erlang/Elixir conventions, NVLang makes \texttt{send} asynchronous by default (fire-and-forget), with explicit \texttt{await} required for synchronous request-reply. This design choice reflects the common case in production BEAM systems where most messages are asynchronous, making the simple case syntactically simple while making blocking communication explicit and intentional.

\textbf{Supervision is a first-class concern.} Most typed actor systems treat supervision as a library feature. NVLang demonstrates that integrating supervision into the type system catches configuration errors that would otherwise manifest as runtime failures.

\textbf{Compilation target matters for adoption.} By targeting Core Erlang rather than introducing a new VM, NVLang leverages decades of BEAM optimization work and enables gradual adoption within existing Erlang codebases.

\textbf{Type inference reduces annotation burden.} Full Hindley-Milner inference allows most NVLang programs to be written without explicit type annotations, lowering the barrier to adoption.

\subsection{Future Work}

Several promising directions emerge from this work:

\textbf{Gradual typing:} Support for gradual typing would enable NVLang code to safely interoperate with untyped Erlang libraries through boundary types and runtime checks.

\textbf{Distributed types:} Extending NVLang's type system to track node locations and network boundaries could provide static guarantees about distributed protocols.

\textbf{Effect systems:} Tracking side effects in the type system could enable more sophisticated program analysis and optimization.

\textbf{Session types:} Complex multi-party protocols could benefit from full multiparty session types as an optional feature.

\subsection{Limitations}

While NVLang provides significant safety guarantees for actor-based programming on the BEAM, several limitations remain:

\paragraph{No Distributed Types} NVLang's type system tracks message protocols locally within a single node. The type system does not currently capture distributed system properties such as network partitions, node failures, or cross-node protocol conformance. When actors communicate across BEAM nodes, the same type safety guarantees apply to message construction, but the type system cannot verify that the remote node is running compatible code.

\paragraph{No Effect Tracking} NVLang does not track side effects in the type system. Operations with external effects---such as file I/O, network calls, or database transactions---are treated identically to pure computations. This means the type system cannot prevent or sequence effectful operations, nor can it guarantee that effects are properly handled or that resources are correctly managed.

\paragraph{Interoperability Boundary} External function calls to Erlang libraries are trusted. The type signatures provided in \texttt{external} declarations are not verified against the actual Erlang implementation. Incorrect type annotations for external functions can lead to runtime type mismatches that violate NVLang's type safety guarantees. This is an inherent limitation of interoperating with dynamically-typed code.

\paragraph{No Session Types} NVLang's typed PIDs capture the set of messages an actor accepts but do not enforce sequencing constraints or protocol state machines. An actor that expects messages in a specific order---for example, requiring \texttt{Connect} before \texttt{Send}---cannot express this constraint in the type system. Clients can send valid messages in invalid orders without triggering compile-time errors.

\paragraph{Polymorphic ADT Limitations} While NVLang supports parametric polymorphism for algebraic data types, higher-kinded type inference is incomplete. Type constructors that abstract over other type constructors---such as functors or monads---require explicit type annotations and cannot be fully inferred. This limits the expressiveness of generic abstractions compared to languages with more sophisticated type systems.
\subsection{Broader Impact}

NVLang addresses a critical gap in the landscape of typed concurrent languages. BEAM powers systems requiring extreme reliability---telecommunications, financial services, real-time messaging platforms. The current choice between Erlang's dynamic typing and abandoning BEAM's unique runtime properties represents a false dichotomy. NVLang demonstrates that static typing and BEAM's actor model are not only compatible but synergistic.

By providing a practical, gradually-adoptable path to typed actor programming on BEAM, NVLang has the potential to reduce bugs in production systems, improve developer productivity through better tooling, and make BEAM's powerful concurrency model accessible to developers who prefer or require static typing.

\bibliographystyle{plain}

\end{document}